\documentclass[11pt]{article}
\usepackage{fullpage}
\usepackage{amsfonts}
\usepackage{amssymb}
\usepackage{amstext}
\usepackage{amsmath}
\usepackage{xspace}
\usepackage{theorem}
\usepackage{graphicx}
\usepackage{graphics}
\usepackage{colordvi}
\usepackage{multirow}
        {\hspace*{\fill}$\Box$\par\vspace{4mm}}
\newcommand{\opt}{\mbox{\sf OPT}}

\newcommand{\be}{\begin{enumerate}}
\newcommand{\ee}{\end{enumerate}}
\newcommand{\bd}{\begin{description}}
\newcommand{\ed}{\end{description}}
\newcommand{\bi}{\begin{itemize}}
\newcommand{\ei}{\end{itemize}}

\newtheorem{lemma}{Lemma}[section]
\newtheorem{fact}{Fact}[section]
\newtheorem{theorem}{Theorem}

\newtheorem{claim}{Claim}[section]

\newtheorem{definition}{Definition}[section]
\newenvironment{proof}{\smallskip \noindent {\bf Proof:}}{\hfill\stopproof}
\def\stopproof{\square}
\def\square{\vbox{\hrule height.2pt\hbox{\vrule width.2pt height5pt \kern5pt
\vrule width.2pt} \hrule height.2pt}}

\newcommand{\ra}{\rightarrow}
\renewcommand{\phi}{\varphi}
\newcommand{\eps}{\epsilon}

\newcommand{\F}{\ensuremath{\mathbb F}}
\newcommand{\E}[1]{\text{\bf E}[#1]}
\newcommand{\Exp}[2]{\text{\bf E}_{#1} #2}
\renewcommand{\Pr}[1]{\text{\bf Pr}\left [#1\right]}

\newcommand{\SPPprob}{{\bf{Sequential Posted Pricing}}\xspace}
\newcommand{\tp}{\tilde{p}}
\newcommand{\tu}{\tilde{u}}
\newcommand{\hu}{\hat{u}}
\newcommand{\p}{\mathbf{p}}

\newcommand{\lpkspp}{\mbox{\sc Lp-K-SPM}}
\newcommand{\dkspp}{\mbox{\sc Dual-K-SPM}}
\newcommand{\lagspp}{\mbox{LagrangianSPM}}

\newcommand{\extgap}{\mbox{VersionGAP}}
\newcommand{\eat}[1]{}

\begin{document}

\title{Approximation Schemes for Sequential Posted Pricing in Multi-Unit Auctions}

\author{Tanmoy Chakraborty \thanks{Part of this work was done while
visiting Google Research. Department of Computer and Information Science. University of Pennsylvania,Philadelphia, PA. Email {\tt tanmoy@cis.upenn.edu}}
\and 
Eyal Even-Dar \thanks{Google Research,
76 Ninth Ave, New York, NY. Email: {\tt evendar@google.com}}
\and Sudipto Guha\thanks{Part of this work was done while visiting
Google Research. Department of Computer and Information Science. University of Pennsylvania,Philadelphia, PA. Email {\tt sudipto@cis.upenn.edu}}
\and Yishay Mansour\thanks{Google Israel and The Blavatnik School of Computer Science, Tel-Aviv University, Tel-Aviv, Israel, Email {\tt mansour.yishay@gmail.com}}
\and S. Muthukrishnan\thanks{Google Research,
76 Ninth Ave, New York, NY. Email: {\tt muthu@google.com}}}

\maketitle

\begin{abstract}
We design algorithms for computing approximately revenue-maximizing {\em sequential posted-pricing mechanisms (SPM)} in {\em $K$-unit auctions}, in a standard {\em Bayesian model}. A seller has $K$ copies of an item to sell, and there are $n$ buyers, each interested in only one copy, who have some value for the item. The seller must post a price for each buyer, the buyers arrive in a sequence enforced by the seller, and a buyer buys the item if its value exceeds the price posted to it. The seller does not know the values of the buyers, but have Bayesian information about them. An SPM specifies the ordering of buyers and the posted prices, and may be {\em adaptive} or {\em non-adaptive} in its behavior.

The goal is to design SPM in polynomial time to maximize expected revenue. We compare against the expected revenue of optimal SPM, and provide a polynomial time approximation scheme (PTAS) for both non-adaptive and adaptive SPMs. This is achieved by two algorithms: an efficient algorithm that gives a $(1-\frac{1}{\sqrt{2\pi K}})$-approximation (and hence a PTAS for sufficiently large $K$), and another that is a PTAS for constant $K$. The first algorithm yields a non-adaptive SPM that yields its approximation guarantees against an optimal adaptive SPM -- this implies that the {\em adaptivity gap} in SPMs {\em vanishes} as $K$ becomes larger.
\end{abstract}

\section{Introduction}
We consider the following \SPPprob problem in a {\em $K$-unit auction}. There is a single seller with $K$ identical copies of a single item to sell, to $n$  prospective buyers. Each buyer is interested in exactly one copy of the item, and has a value for it that is unknown to the seller. The buyers arrive in a sequence, and each buyer appears exactly once. The arrival order may be chosen by the seller. The seller quotes a price for the item to each arriving buyer, and may quote different prices to different buyers. Assuming that buyers are rational, a buyer buys the item if the price quoted to him is less than his value for the item, and pays the quoted price to the seller. This process stops when either $K$ buyers have bought the item or when all buyers have arrived and left.

We focus on {\em pricing and ordering strategies} in the above model, called {\em sequential posted-price mechanisms ({\bf SPM}s)},  that maximize the seller's expected revenue. Posted price mechanisms are clearly incentive compatible, and commonly used in practice. We design strategies in a {\em Bayesian framework}, where each buyer draws his value of the item from a distribution.  These value distributions are known to the seller, and are used in designing the mechanism.

SPMs were recently studied in the general context of {\em Bayesian single-parameter mechanism design (BSMD)}, which includes our $K$-unit auction, by Chawla et. al. \cite{CHMS10}. They designed efficiently computable SPMs for various classes of BSMD problems and compared their expected revenue to that of the optimal auction mechanism, which was given by Myerson \cite{M81}. For the $K$-unit auction, they showed that their SPM guarantees $(1-1/e)$-fraction of the revenue obtained by Myerson's auction. Bhattacharya et. al. \cite{BGGM10} (as well as \cite{CHMS10}) also used sequential item pricing to approximate optimal revenue, when the seller has multiple distinct items. However, the SPM computed by their algorithms may {\em not} be the optimal SPM, {\em i.e.} there may exist SPMs with greater expected revenue. Given that SPMs are quite common in practice, we focus in this paper on efficiently computing an optimal SPM.

\paragraph{Our Results}
The results in \cite{CHMS10} immediately imply a $(1-1/e)$-approximation for the problem of computing optimal SPMs in $K$-unit auction. We strictly improve this bound. We design two different algorithms -- the first is a polynomial time algorithm that gives $(1 - \frac{1}{\sqrt{2\pi K}})$-approximation, and is meant for large values of $K$, and the second is a polynomial time approximation scheme (PTAS) for constant $K$. Combining these two algorithms yield a polynomial time approximation scheme for the optimal SPM problem, for all values of $K$: if $K>\frac{1}{2\pi \eps^2}$, run the first algorithm, else run the second algorithm. Recall that a PTAS is an algorithm that, for any given constant $\eps>0$, yields $(1-\eps)$-approximation in polynomial time (the exponent of the polynomial should be a function of $\eps$ only, and independent of input size).

Note that a sequential posted pricing strategy can be adaptive -- it can alter its prices and the ordering of the remaining buyers based on whether the current buyer buys the item. We shall call such strategies as {\em Adaptive SPMs, or {\bf ASPM}s}, while SPM shall refer to a non-adaptive pricing and ordering strategy. Clearly, the expected revenue from an optimal ASPM is at least that from an optimal SPM. Our first algorithm outputs an SPM, but our proof shows that it gives the same approximation guarantee of $(1 - \frac{1}{\sqrt{2\pi K}})$ against an optimal ASPM.  This yields a corollary that the {\em adaptivity gap asymptotically vanishes as $K$ increases}. On the other hand, it is easy to construct instances with $K=2$, such that there is a constant factor {\em adaptivity gap}, {\em i.e.} gap in expected revenue between optimal SPM and ASPM. We design a third algorithm that outputs an ASPM, and is a PTAS for computing an optimal ASPM, for constant $K$. Again, combining this result with our first algorithm, we obtain a PTAS for the optimal ASPM problem, for all values of $K$. Adaptive PTAS with multiplicative approximation is rare to find in stochastic optimization problems. For example, an adaptive PTAS for the stochastic knapsack problem has been developed very recently \cite{BGK10}. The theorem below summarizes our results.

\vspace*{-0.0cm}\begin{theorem}
There is a PTAS for computing a revenue-maximizing SPM in $K$-unit auctions, for all $K$. The same result holds for ASPMs.
\end{theorem}

\paragraph{Our Techniques}
The first algorithm is based on a linear programming (LP) relaxation of the problem, such that the optimal solution to the LP upper bounds the expected revenue from any ASPM. We show that this LP has an optimal integral solution, from which we construct a pricing for the buyers. The buyers are ordered simply in decreasing order of prices -- it is easy to see that this is an optimal ordering policy given the prices. The LP formulation implies that if there were no limit on the number of copies the seller can sell, then the expected revenue obtained from this pricing would be equal to the LP optimum, and {\em at most $K$ copies of the item are sold in expectation}. However, the algorithm is restricted to selling at most $K$ copies in all realizations, and the result follows by bounding the loss due to this hard constraint. The interesting property we find is that {\em this loss vanishes as $K$ increases}. It should be noted that an LP-based approach is used in \cite{BGGM10}; however, they consider a more general problem with multiple distinct items, and their analysis yielded no better than constant approximation factors.

The second algorithm uses a dynamic programming approach, which is common in the design of approximation schemes. We make some key observations that reduce the problem to an {\em extended version of the generalized assignment problem (GAP) \cite{ST93,CK05} with constant number of bins}, which has polynomial time algorithm (polynomial in the size of bins and number of items) using dynamic programming. The main observation is that in any SPM, if we pick a contiguous subsequence of buyers to whom there is very small probability of selling even a single copy, and arbitrarily permute this subsequence, the resulting SPM will have almost the same expected revenue as the original SPM. This observation drastically cuts down the number of configurations that we have to check before finding a near-optimal SPM.

The third algorithm for computing ASPM is a generalization of the second algorithm, but it must now approximate a decision tree, that may branch at every step based on whether a copy is bought, instead of an SPM sequence. The key observation in this case is that there exists a near-optimal decision tree that does not branch too often, and the problem again reduces to an extension of GAP with constant number of bins.

\paragraph{Other Related Work}
Sequential item pricing for combinatorial auctions has also been studied in prior-free settings, where no knowledge about the buyers' valuation is assumed ({\em eg.} \cite{BBM08,CHK09}). These results compare the revenue obtained to the optimal social welfare, primarily due to lack of a better upper bound, and get no better than logarithmic approximation results. Maximizing welfare via truthful mechanisms in prior-free settings have been studied for $K$-unit auctions \cite{DD09,DN07} and other combinatorial auctions \cite{DNS06,F06}. Bayesian assumptions provide better upper bounds, and has led to constant approximation against optimal revenue for any auction \cite{BGGM10,CHMS10}. But Bayesian assumptions can lead to tighter upper bounds on optimal sequential pricing, and that is our main contribution. A parallel posted-price approach has been used in a more complex repeated ad auction setting to get constant approximation \cite{CEGMM10}.

\section{Preliminaries}\label{prelim}

In a $K$-unit auction, there is a single seller who has $K$ identical copies of a single item, and wish to sell these copies to $n$ prospective buyers $B_1,B_2\ldots B_n$. Each buyer $B_i$ is interested in one copy of the item, and has value $v_i$ for it. $v_i$ is drawn from a distribution specified by cumulative distribution function (cdf) $F_i$ that is known to the seller. The values of different buyers are {\em independently} drawn from their respective distributions. Without loss of generality, we assume that $K\leq n$.

\vspace*{-0.0cm}\begin{definition}
Let $\p_{iv}$ denote the probability that $B_i$ has value $v$ for the item. Let $\tp_{iv}$ denote the probability that $B_i$ has value at least $v$. We shall call it the {\em success probability} when $B_i$ is offered price $v$. Clearly $\tp_{iv} = \sum_{v' \geq v} \p_{iv'}$. 
\end{definition}

We assume, for all our results, that each value distribution is discrete, with at most $L$ distinct values in its {\em support} ({\em i.e.} these values have non-zero probability mass). Let $U_{V_i}$ be the support set of values for the distribution of $B_i$, and let $U_V=\bigcup_{i=1}^n U_{V_i}$. We shall also assume that $L$ is polynomial in $n$, and that $\tp_{iv}$ is an integral multiple of $\frac{1}{10n^2}$ for all $i,v$. These assumptions are without loss of generality for obtaining PTAS for optimal SPM or ASPM (see Appendix \ref{discrete} for a brief discussion).

\vspace*{-0.0cm}\begin{definition}
A sequential posted-price mechanism (SPM) is a mechanism which considers buyers arrive in a sequence, and offers each of them a take-it-or-leave-it price: the buyer may either buy a copy at the quoted price or leave, upon which the seller  makes an offer to another buyer. Each buyer is given an offer at most once, and the process ends when either all $K$ copies have been sold, or there is no buyer remaining. 

An SPM specifies the entire sequence of buyers and prices before the process begins. In contrast, an adaptive sequential posted-price mechanism (ASPM) may decide the next buyer based on which of the current and past buyers accepted their offered prices. 
\end{definition}

Note that there can be no adaptive behavior when $K=1$, since the process stops with the first accepted price.
Thus an ASPM can be specified by a {\em decision tree}: each node of the tree contains a buyer and a price to offer. Each node may have multiple children. The selling process starts at the root of the tree ({\em i.e.} offers the price at the root to the buyer at the root), and based upon whether a sale occurs at the root, moves to one of the children of the root, and continues inductively. The process stops when either $K$ items have been sold, or $n$ buyers have appeared on the path in the decision tree traversed by the process -- the latter nodes are the leaves of the decision tree.

 It is easy to see that the decision of an optimal ASPM at any node of the tree should depend only on the number of copies of the item left and the remaining set of buyers (the latter is solely determined by the node reached by the process). Thus, each node has at most $K$ children, at most one each for the number of copies left. Note that an ASPM may not adapt immediately to a sale -- it may move to a fixed buyer regardless of the outcome. Such a node will only have a single child. Without loss of generality, we shall represent an ASPM such that each non-leaf node either has a single child or $K$ children (some of which may even be infeasible). The latter nodes are called {\em branching nodes}. In this context, an SPM is simply an ASPM whose decision tree is a path.

SPM and ASPM are incentive compatible: a buyer $B_i$ buys the item if and only if its value $v_i$ is equal to or greater than the price offered to it, and pays only the quoted price to the seller. 

\vspace*{-0.0cm}\begin{definition}
The revenue $R(v_1,v_2\ldots v_n)$ obtained by the seller for a given SPM is the sum of the payments made by all the buyers, which is a function of the valuations of the buyers. The expected revenue of an SPM or ASPM is computed over the value distributions $\Exp{v_i\sim F_i}{R(v_1,v_2\ldots v_n)}$. An {\em optimal} SPM or ASPM is an SPM (respectively, ASPM) that gives the highest expected revenue among all SPMs (respectively, ASPMs). 

Let the expected revenue of an optimal SPM (or ASPM) be $\opt$. An $\alpha$-approximate SPM (or ASPM, respectively), where $\alpha\leq 1$, has expected revenue at least $\alpha \opt$.
\end{definition}
%\paragraph{Discretization}

%This is without loss of generality, since we can discretize each distribution suitably while losing at most $1/n$ of the revenue. First, discretize the set of values to powers of $(1-1/n^2)$, and construct a new {\em intermediate} distribution by defining a cdf that is equal to the original cdf at these points, and constant between these points. Given an SPM (or ASPM), its prices can be altered to their respective closest discrete points, and the expected revenue can only change by a factor of $(1-1/n^2)$ due to this transformation. 

%Next, consider the discretized support in the intermediate distribution, and select a maximal subset of them such that the cdf at any pair in the subset differ by at least $1/n^2$. There are at most $n^2$ points in this subset. Define a modified cdf that matches the previous cdf at all points in the subset, and is constant between these points. Then the selected subset of points form the support of this new distribution. Again, given an SPM (or ASPM) with prices from the intermediate distribution, each price in the SPM can be mapped to the least price exceeding it that belongs to the selected subset. In the modified SPM, the probability of reaching a particular buyer before all items are sold, or the probability that it accepts the offered price, do not change by more than $1/n$

\subsection{Basic Result}\label{basic-results}
An SPM must specify an ordering of the buyers as well as the prices to offer to them. It is worth noting that if either one of these tasks is fixed, the other task becomes easy. 
%(Please see Appendix \ref{missing-proofs} for missing proofs of lemmas.)
\vspace*{-0.0cm}\begin{lemma}\label{basic-easy}
Given take-it-or-leave-it prices to offer to the buyers, a revenue-maximizing SPM with these prices simply considers  buyers in the order of decreasing prices. Given an ordering of buyers, one can compute in polynomial time a revenue-maximizing ASPM that uses this ordering (and only adapts the offered prices).
\end{lemma}
\begin{proof}
For the first claim, consider an SPM where there are two buyers $B_i$ and $B_j$, such that $B_i$ arrives just before $B_j$, but is offered a lower price than $B_j$. Consider the modified SPM created by swapping $B_i$ and $B_j$ in the order, while keeping the price offered to $B_i, B_j$ and other buyers unchanged. In realizations where at most one of $B_i$ or $B_j$ accepts the price, the revenue of the original and modified SPMs are equal. However, in realizations where both buyers accept their offered prices, the selling process may not reach the latter buyer, and so the modified SPM has higher or equal revenue in that case.

For the second claim, we can compute the prices using dynamic programming. Let the buyers be ordered as $B_{\pi(1)},B_{\pi(2)}\ldots B_{\pi(n)}$. Let $A(i,j)$ denote the maximum expected revenue that can be obtained from the {\em last} $i$ buyers in the given ordering, if there are $j$ items left to sell to them. For initialization, set $A(1,0)=0$, and $A(1,j)=\max_{x\in U_V} x\Pr{v_n\geq x}$ for $j\geq 1$, which is the maximum expected revenue from $B_n$ with the item in stock. Suppose $A(i-1,j)$ has been computed for all $j$. Then $A(i,j)$ can be computed by iterating through all possible prices to offer $B_{n-i+1}$, and pick one that yields highest expected revenue. For a price $x\in U_V$, the expected revenue is $(x + A(i-1,j-1))\Pr{v_{n-i+1}\geq x} + A(i-1,j)\Pr{v_{n-i+1}< x}$. Finally, $A(n,K)$ is the expected revenue. Note that we could store, along with $A(i,j)$, the optimal price $p$, and these prices yield the required ASPM.
\end{proof}
%\paragraph{Extended GAP}

\section{LP-based Algorithm for Large $K$}\label{sec-LPalgo}

In this section we present our first algorithm that yields us an approximation factor that improves as $K$ increases, and implies a vanishing adaptivity gap. The following theorem summarizes our result.

\vspace*{-0.0cm}\begin{theorem}\label{thm:LPalgo}
For all $K \geq 1$, if a seller has $K$ units to sell, there exists an SPM whose expected revenue is at least $1 - \frac{K^K}{K! e^K} \geq 1 - \frac{1}{\sqrt{2\pi K}}$ fraction of the optimal ASPM. This SPM can be computed in polynomial time.
\end{theorem} 

As a first step to our algorithm, we add random infinitesimal perturbation to the values $v\in U_{V_i}$ and the associated probability values $\p_{iv}$, so that almost surely, $U_{V_i}$ are disjoint, and further, all the values and probabilities are in {\em general position}. Intuitively, this property is used in our algorithm to break ties.

Consider any ASPM $\mathcal{P}$, that may even be randomized. Consider the event $E_{iv}$ that $B_i$ is offered the item at price $v$, and accepts the offer. Let $y_{iv}$ denote the probability of that $E_{iv}$ occurs when $\mathcal{P}$ is implemented. Let $x_{iv}$ denote the probability that $B_i$ was offered price $v$ when $\mathcal{P}$ is implemented.  Note that both probabilities are taken over the value distributions of the buyers, as well as internal randomization of $\mathcal{P}$. Naturally, we must have $y_{iv}  \leq \tp_{iv} x_{iv}$. Also, by linearity of expectation, the expected revenue obtained by $\mathcal{P}$ is $\sum_{i=1}^n  \sum_{v\in U_V} v y_{iv}$. Moreover, $\sum_{i=1}^n  \sum_{v\in U_V}  y_{iv}$ is the expected number of copies of the item sold by the seller, and this quantity must be at most $K$. Finally, the mechanism enforces that  each buyer is offered a price at most once in any realization, and hence in expectation,{\em i.e.}  $\sum_{v\in U_V} x_{iv}\leq 1$.

Viewing $x_{iv}$ and $y_{iv}$ as variables depending upon the selected ASPM, optimum of the following linear program $\lpkspp$ provides an upper bound to the expected revenue from any ASPM, since any ASPM provides feasible assignment to the variables. Our algorithm involves computing an optimal solution to this program with a specific structure, and use the solution to construct an SPM. We also consider its dual program, $\dkspp$.

\begin{tabular}{c|c}
\begin{minipage}{3in}
{\small 
\begin{eqnarray*}
& \lpkspp = & \max \sum_{i=1}^n \sum_{v\in U_V} v y_{iv}  \\
& y_{iv} & \leq \tp_{iv} x_{iv}\ \ \ \ \forall i\in[1,n],v\in U_V\\
& \sum_{v\in U_V}  x_{iv} & \leq 1\ \ \ \ \ \ \ \ \ \ \forall i\in [1,n]\\
& \sum_{i=1}^n \sum_{v\in U_V} y_{iv} & \leq K \\
& & y_{iv},x_{iv} \geq 0
\end{eqnarray*}
}
\end{minipage}
& 
\begin{minipage}{3in}
{\small
\begin{eqnarray*}
& \dkspp = & \min  K \tau + \sum_i \lambda_i \\
& & \zeta_{iv} + \tau \geq v \\
& & \lambda_i - \sum_v \tp_{iv} \zeta_{iv} \geq 0 \\
& & \zeta_{iv},\lambda_i,\tau \geq 0
\end{eqnarray*}
}
\end{minipage}
\end{tabular}
%Let the optimum solution of the above be given by
%$\{y^*_{iv},x^*_{iv}\zeta^*_{iv},\lambda^*_i,\tau^*\}$.

%The proof of the following lemma is in Appendix \ref{missing-proofs}.
\vspace*{-0.0cm}\begin{lemma}\label{structured}
Assuming that the points in $U_{V_i}$ and the probabilities $\tp_{iv}$ have been perturbed infinitesimally, and so are in general position, there exists an optimal {\em structured solution} $x^*_{iv}, y^*_{iv}$ of $\lpkspp$, computable in polynomial time, such that:
\begin{enumerate}
\item for all $i,v$, $y_{iv}=\tp_{iv} x_{iv}$.
\item for each $i$ there is exactly one $v$ such that $x_{iv}>0$. Let $v(i)$ denote the value for which $x_{iv(i)}>0$.
\item There exists at most one $i$ such that $1\leq i\leq n$ and $x_{iv(i)}=1$. If such $i=i'$ exists, then $v(i')=\min_{i=1}^n v(i)$. 
\end{enumerate}
\end{lemma}
\begin{proof}
Given any feasible solution to $\lpkspp$, where $y_{iv} < \tp_{iv} x_{iv}$ for some $i,v$, we can simply reduce $x_{iv}$ till $y_{iv}$ becomes equal to $\tp_{iv}x_{iv}$. This change keeps the solution feasible, and also leaves the objective unchanged. So we can simply eliminate the variables $y_{iv}$ from $\lpkspp$ by setting $y_{iv}=\tp_{iv} x_{iv}$. An optimal solution to this modified LP will also be an optimal solution for the original LP, and naturally satisfy the first condition in the lemma. By a minor abuse of notation, we refer to this modified LP as $\lpkspp$.

Let us now consider the Lagrangian program $\lagspp(\tau)$ obtained by removing the constraint of selling at most $K$ copies, and associating a cost $\tau$ of violating this constraint in the objective.  The following property holds by LP duality: let $\tau^*$ be the assignment to variable $\tau$ in an optimal solution to $\dkspp$. Then the optimum of $\lagspp(\tau^*)$ is equal to the optimum of $\lpkspp$ in value. We shall compute an optimal solution of $\lagspp(\tau^*)$ that is also feasible for $\lpkspp$, and satisfies either $\tau^*=0$ or $\sum_i \sum_v \tp_{iv}x_{iv} = K$. Such a solution must also be an optimal solution of $\lpkspp$.

%\begin{tabular}{c|c}
%\begin{minipage}{3in}
\vspace*{-0.0cm}
{\small 
\begin{eqnarray*}
&\hspace*{-0.5cm} \lagspp(\tau)  =& \hspace*{-0.2cm} \max \left(\sum_i \sum_{v} v \tp_{iv} x_{iv} +\tau(K - \sum_i \sum_v \tp_{iv}x_{iv})\right) = K \tau + \sum_i \max \sum_{v} (v-\tau) \tp_{iv} x_{iv}  \\
& &\sum_v x_{iv} \leq 1 \\
& &y_{iv},x_{iv} \geq 0
%\end{center}
\end{eqnarray*}
}
\vspace*{-0.0cm}
%\end{minipage}
%& 
%\begin{minipage}{3in}
%{\small
%\begin{eqnarray*}
%& \dlagspp(\tau) = & K \tau + \sum_i \min \lambda_i \\
%& & \zeta_{iv} \geq v - \tau\\
%& & \lambda_i - \sum_v \tp_{iv} \zeta_{iv} \geq 0 \\
%& & \zeta_{iv},\lambda_i,\tau \geq 0
%\end{eqnarray*}
%}
%\end{minipage}
%\end{tabular}

%Let a pair of optimal solutions to the Lagrangian LP and its dual be denoted by $\{y^*_{iv}(\tau),x^*_{iv}(\tau),\zeta^*_{iv}(\tau),\lambda^*_i(\tau)\}$.

%By Compliementary slackness, for some $v > \tau$ which can be a
%possible value of $i$, we must have $\zeta^*_{iv}(\tau) > 0$ and as a
%consequence we must have $y^*_{iv}(\tau)=\tp_{iv} x^*_{iv}(\tau)$. For $v
%<\tau$ we will have $y^*_{iv}(\tau)=0$. 
%

%Note that $\tau=v$ for at most
%one bidder, by assumption of the perturbations -- moreover it does not
%affect the oblective function of $\lagspp(\tau)$. Therefore

%{\small 
%\begin{eqnarray*}
%& \lagspp(\tau) = & K \tau + \sum_i \max \sum_{v > \tau} (v-\tau) \tp_{iv} x^*_{iv}(\tau)  \\
%& & \sum_v x^*_{iv}(\tau) \leq 1 \\
%& & x^*_{iv}(\tau) \geq 0
%\end{eqnarray*}
%}

If $\tp_{iv}=0$, then we assume that $x_{iv}$ is set to zero, since this does not affect feasibility or value of the objective. Let an optimal solution of $\lagspp(\tau^*)$ be denoted by $x_{iv}^*(\tau^*)$. Such a solution must satisfy $x_{iv}^*(\tau^*)=0$ for all $v<\tau^*$. Further for some $i$, if $\max_v (v-\tau^*) \tp_{iv}$ is maximized at a unique $v$, then $x_{iv}^*(\tau^*)=1$ if $v= \arg \max_{v\in U_{V_i}} \{ (v-\tau^*) \tp_{iv} | v \geq \tau^* \}$, and $0$ otherwise. If $\max_v (v-\tau^*) \tp_{iv} > 0$, then the added perturbations ensure that the maximum is indeed unique.

Suppose the maximum is zero for some $i$, then $\tp_{iv}=0$ and so $x_{iv}^*(\tau^*)=0$ $\forall v> \tau^*$. Since every buyer has some non-zero probability of having a positive value for the item (else we can simply neglect such buyers), we have $\tau^*>0$. The only $x_{iv}^*(\tau^*)$ that we may set to a non-zero value is for $v=\tau^*$, provided that $\tau^*\in U_{V_i}$. Because of added perturbations, this can happen for at most one buyer $i$. We first fix the assignment of all other variables as described above, then set this $x_{iv}^*$ to the highest value less than $1$ such that $\sum_i \sum_v \tp_{iv}x_{iv}) \leq K$. This gives our required structured solution. Given $\tau^*$, constructing the solution requires {\em linear time}. As mentioned before, $\tau^*$ can be computed by solving $\dkspp$; one may also use binary search techniques for this purpose, similar to many packing LPs (details omitted, see, {\em eg.} \cite{R96}).
\end{proof}

Our algorithm for computing an SPM is as follows:
\begin{enumerate}
\item Compute an optimal structured solution of $\lpkspp$.
\item In the SPM, offer price $v(i)$ to $B_i$, and consider buyers in order of decreasing $v(i)$.
\end{enumerate}

\subsection{Approximation Factor}

It remains to analyze the approximation factor of our algorithm. Let the order of decreasing prices be $B_{\pi(1)},B_{\pi(2)}\ldots B_{\pi(n)}$. For $1\leq i<n$, let $Z_i$ be a two-valued random variable that is $v(\pi(i))=z_i$ with probability $\tp_{\pi(i) v(\pi(i))}=u_i$, and $0$ otherwise. To define $Z_n$, note that $x_{\pi(n)v(\pi(n))}^*$ in the structured optimal solution may not have been $1$, so let $Z_n$ be $v(\pi(n))=z_n$ with probability $x_{\pi(n)v(\pi(n))}^* \tp_{\pi(n) v(\pi(n))}=u_n$ and $0$ otherwise. If $Z=\sum_{i=1}^n Z_i$, then $\E{Z}$ is the optimum of the LP solution. The revenue of the algorithm, however, is at least equal to the sum of the first $K$ variables in the sequence $Z_1,Z_2\ldots Z_n$ that are non-zero. Let this sum be denoted by the random variable $Z'$. Note that $z_1\geq z_2 \geq \ldots \geq z_n$, and $\sum_{i=1}^n u_i\leq K$. The following lemma immediately implies Theorem \ref{thm:LPalgo}.

\vspace*{-0.0cm}\begin{lemma}\label{approxfactor}
$\E{Z'}\geq (1 - \frac{K^K}{K! e^K})\E{Z} \geq (1 - \frac{1}{\sqrt{2\pi K}})\E{Z}$.
\end{lemma}

%
%\paragraph{Sequential Decision Making:} To use the LP bound we have to use a sequential process which can be 
%emulated by the posted price scheme. To this end:
%\begin{enumerate}
%\item We order the buyers based on decreasing $v(i)$.
%\item For buyer $i$ define $Y_i$ to be expected revenue from the buyer. 
%$Y_i = \left\{ \begin{array}{ll}
%        0 & \mbox{if $X_i < v(i)$}\\
%        v(i) & \mbox{Otherwise, with prob. $u(i)$}
%\end{array}
%\right.
%$.
%\item We now play the game that we offer the prices to the buyers till we run out of buyers or we sell $K$ items.
%\end{enumerate}
%
%
%\vspace*{-0.0cm}\begin{lemma}
%Suppose we play the following game: We are given a set of random
%variables with the probability $\Pr[Y_i \neq 0] = u(i)$ and $\E[Y_i|
%  Y_i \neq 0]= v(i)$.  Suppose which we are inspecting the variables
%in increasing order and the game stops when we have $K$ nonzero values
%(or exhaust the ordering). Let the random vriable $Y$ denote the sum of the 
%values seen. Then 
%\[ \E[Y] \geq \left( 1 - \varphi(K) \right) \sum_i v(i) u(i) \]
%where $\varphi(K) = \frac{K^K}{K! e^K} \rightarrow \frac{1}{\sqrt{2\pi K}}$.
%\end{lemma}
\vspace*{-0.0cm}\begin{proof}
Let $\alpha(i)=z_i u_i$. Let the probability that we reach $Z_i$ in the sequence before finding $K$ non-zero variables, be given by the function $f(i,\vec{u})$ (this function is independent of $z_1,z_2\ldots z_n$), where $\vec{u}=(u_1,u_2\ldots u_n)$. Then $\E{Z'}=\sum_{i=1}^n 
f(i,\vec{u}) \alpha(i)$, while $\E{Z}=\sum_{=1}^n \alpha(i)$. Observe that $f(i,\vec{u})$ is monotonically decreasing in $i$.
We shall narrow down the the instances on which $\E{Z'}/ \E{Z}$ is minimized.

\vspace*{-0.0cm}\begin{claim}\label{bernoulli}
Given an instance comprising variables $Z_1,Z_2\ldots Z_n$ such that $z_i>z_{i+1}$, one can modify it to construct another instance $\tilde{Z_1},\tilde{Z_2}\ldots\tilde{Z_n}$ such that $\E{Z'}/ \E{Z}$ decreases.
\end{claim}
\begin{proof}
We modify the instance by defining $\alpha'(j)$ for $1\leq j\leq n$, and setting the possible non-zero value of $\tilde{Z_j}$ to be $\tilde{z_j}=\alpha'(j)/u_j$, with success probability remaining $u_j$:
{\small
\[
\alpha'(j) = \left \{
\begin{array}{ll}
\alpha(j) & \mbox{ if $j \neq i,i+1$}\\
\alpha(i) - \Delta & \mbox{ if $j = i$}\\
\alpha(i+1) + \Delta & \mbox{ if $j = i+1$}
\end{array} \right.
\quad \quad \mbox{where}
\quad \quad \Delta = \frac{z_i - z_{i+1}}{
  \frac1{u(i)} +\frac1{u(i+1)}} > 0 \]
}

Note that $\tilde{Z_j}=Z_j\ \forall j\neq i,i+1$, so only $Z_i$ and $Z_{i+1}$ gets modified. Further, $\tilde{z_j}$ are non-increasing in $j$ (in fact, $\tilde{z_i}=\tilde{z_{i+1}}$) so the modified instance is valid. Also, $\sum_i \alpha(i) =\sum_i \alpha'(i)$, so $\E{Z}$ remains unchanged. $\vec{u}$ remains unchanged too, and hence the probabilities $f(i,\vec{u})$. Finally, the change in $\E{Z'}$ is $\left(f(i+1,\vec{u})-f(i,\vec{u})\right) \Delta < 0$, {\em i.e.} $\E{Z'}$ decreases.
\end{proof}

%See Appendix \ref{missing-proofs} for a proof of the claim. 
Thus, we can restrict our attention to instances where $z_1=z_2=\ldots =z_n =z^*$ (say). Without loss of generality, we let $z=1$, so that $Z_1,Z_2\ldots$ are Bernoulli variables, and $Z' = \min\{Z,K\}$. Note that the ordering of the variables do not influence $Z'$. The next step is to show that if we {\em split} the variables, keeping $\E{Z}$ unchanged, $\E{Z'}$ can only decrease.
% (proof of the claim is in Appendix \ref{missing-proofs}).

\vspace*{-0.0cm}\begin{claim}\label{iid}
Let $Z_1,Z_2\ldots Z_n$ be Bernoulli variables, such that the {\em success probability} is $\Pr{Z_j=1}=u_j$. Suppose that we modify the set of variables by removing $Z_i$ from it and adding two Bernoulli variables $\tilde{Z}_i$ and $\hat{Z}_i$ to it, where $\Pr{\tilde{Z}_i=1}=\tu_i > 0$ and $\Pr{\hat{Z}_i=1}=\hu_i > 0$, and $\tu_i + \hu_i = u_i$. Then $\E{Z'}=\E{\min{Z,K}}$ decreases or remains unchanged due to this modification, while $\E{Z}=K$ remains unchanged.
\end{claim}
\begin{proof}
Let $X$ be the sum of the remaining variables, {\em i.e.} $X=\sum_{j=1}^{i-1}Z_j + \sum_{j=i+1}^n Z_j$. We shall show that $\E{Z'|X\leq K-2}$ and $\E{Z'|X\geq K}$ remain unchanged by the modification, while $\E{Z'|X=K-1}$ decreases, thus proving .

If $X\leq K-2$, then $Z'$ is $X + Z_i$ in the original instance, and $X+\tilde{Z}_i + \hat{Z}_i$ in the modified instance. Since $\E{Z_i} = \E{\tilde{Z}_i + \hat{Z}_i}= u_i$, so $\E{Z'|X\leq K-2}$ remains unchanged. Also, if $X\geq K$, then $Z'$ is simply $K$ in both instances. If $X=K-1$, then $Z' = K-1 + Z_i$ in the original instance and $Z'=K-1 + \min\{1,\tilde{Z}_i + \hat{Z}_i\}$ in the modified instance. So $\E{Z'|X=K-1} = K-1+u_i$ and $\E{Z'|K-1} = K-1 + \Pr{\tilde{Z}_i + \hat{Z}_i \geq 1} = K-1 + (\tu_i + \hu_i - \tu_i\hu_i) < K-1+u_i$, respectively.
\end{proof}

Assume that the success probabilities of the Bernoulli variables are all rational -- since rational numbers form a dense set in reals, this shall not change the lower bound we are seeking. Then, there exists some large integer $N$ such that all the probabilities are integral multiples of $1/N$. Further, we can choose an arbitrarily large $N$ for this purpose. Now, split each variable that has success probability $t/N$ into $t$ variables, each with success probability $1/N$. The above claim implies that $\E{Z'}/\E{Z}$ can only decrease due to the splitting. Thus, it remains to lower bound $\E{Z'}/K$ for the following instance, as $N\ra \infty$: $KN$ Bernoulli variables, each with success probability $1/N$.

For this final step, we use the well-known property that the sum of Bernoulli variables with infinitesimal success probabilities approach the {\em Poisson} distribution with the same mean. In particular, if $P$ is a Poisson variable with mean $K$, then the total variation distance between $Z$ and $P$ is at most $(1-e^{-K})/N$ (see {\em e.g.} \cite{BH84}), which tends to zero as $N\ra \infty$. Thus, we simply need to find $\E{\min{P,K}}/K$, and this is the lower bound on $\E{Z'}/\E{Z}$ that we are seeking. It can be verified that $\E{\min{P,K}}= K(1 - \frac{K^K}{K!e^K})$ (see Appendix \ref{poisson}), which proves the lemma.
\end{proof}

\section{PTAS for constant K}\label{ptas-small}

We now define an optimization problem called $\extgap$, and our PTAS for both SPM and ASPM for constant $K$ will reduce to solving multiple instances of this problem.

\medskip
\noindent {\bf $\extgap$:} Suppose there are $n$ objects, and each object has $L$ {\em versions}. Let version $j$ of object $i$ have profit $p_{ij}$ and size $s_{ij}\leq 1$. Also, suppose there are $C$ bins $1,2\ldots C$, where bin $\ell$ has size $s_\ell$ and a {\em discount factor} $\gamma_\ell$ . The goal is to {\em place} versions of objects to bins, such that:
\begin{enumerate}
\item Each object can be placed into a particular bin at most once, as a unique version. If object $i$ is placed as version $j$ into bin $\ell$, then it realizes a profit of $\gamma_\ell p_{ij}$ and a size of $s_{ij}$. 
\item Each object can appear in multiple bins, as different versions. However, there is a given collection $\mathcal \F_C$ of {\em feasible subsets} of bins $1,2\ldots C$. The set of bins that an object is placed into must be a feasible subset.
\item The sum of realized sizes of objects placed into any bin $\ell$ must be less than $s_\ell$.
\end{enumerate}
The profit made by an assignment of object version to bins, that satisfy all the above conditions, is the sum of realized profits by all objects placed in the bins. The goal is to find an assignment that maximizes the profit.
%The generalized assignment problem (GAP) \cite{ST93,CK05} is a special instance of $\extgap$, where each item has a single version, and can be placed into at most one bin. Following is the main lemma we shall use about $\extgap$.
\vspace*{-0.0cm}\begin{lemma}\label{lem-extgap}
For all objects and versions $i,j$, let $s_{ij}$ be a multiple of $1/M$ for some fixed $M\geq 2$. Then an optimal solution to $\extgap$ can be found in time $(ML)^{O(C)} n$.
\end{lemma}
\begin{proof}
The algorithm is a simple dynamic programming. Order the objects arbitrarily. Let $D(i,j_1,j_2\ldots j_C)$ be an optimal feasible assignment (or the profit thereof, by an abuse of notation) of the first $i$ objects, such that the sum of realized sizes of objects in bin $\ell$ is $j_\ell$, for $\ell=1,2\ldots C$. $D(i,j_1,j_2\ldots j_C)$ is assigned as {\em null} and its profit as $-\infty$ if no such assignment exists. Note that we only consider $j_\ell$ to be multiples of $1/M$ and at most $1$, for all $\ell$.

$D(0,j_1,j_2\ldots j_C)$ is null for all $j_1,j_2\ldots j_C$, except for $D(0,0,0\ldots 0)$ which is zero. Suppose $D(i-1,j_1,j_2\ldots j_C)$ have been computed for all $j_1,j_2\ldots j_C$. Then to compute $D(i,j_1\ldots j_C)$, we first choose a feasible subset of bins from $\mathcal{F}_C$ to place it in ($|\mathcal{F}_C| < 2^C$ choices), then its version in each bin in this subset (at most $L^C$ choices), and then compute the objective as $D(i-1,j_1-s_{i t_1},j_2-s_{i t_2}\ldots j_C - s_{i  t_C}) + \sum_{\ell=1}^C \gamma_\ell p_{i t_\ell}$, where $t_\ell$ is the version in which object $i$ is chosen to be placed in bin $\ell$ (if the object is not placed in bin $\ell$, treat $s_{i t_\ell}$ and $p_{i t_\ell}$ as zero).

We iterate through all the choices to maximizes this objective. Thus, computing each entry $D(i,j_1\ldots j_C)$ takes time at most $O(C(2L)^C)$. The number of entries is at most $n M^C$. The maximum among all the entries gives the required assignment.
\end{proof}

\subsection{PTAS for Computing SPM}\label{sec-ptas-spm}
We now design an algorithm to compute a near-optimal SPM for constant $K$. 
%The algorithm also produces prices that will give expected revenue close to $\opt$, but prices which are as good can also be computed after the algorithm produces the order, by using our dynamic programming-based algorithm described in the proof of Theorem \ref{basic-easy}.

\vspace*{-0.0cm}\begin{theorem}\label{ptas-spm}
There exists a PTAS for computing an optimal SPM, for any constant $K$. The running time of the algorithm is $\left(\frac{nk}{\eps}\right)^{poly(k, \eps^{-1})}$, and gives $(1-\eps)$-approximation.
\end{theorem}

We shall, without loss of generality, give a $(1-ck\eps)$-approximation, and this will imply the above theorem: putting $\eps=\eps'/ck$ will yield a $(1-\eps')$-approximation. 

We first establish some definitions that we shall use. Let a {\em segment} refer to a sequence of some buyers and prices offered to these buyers -- we shall refer to parts of an SPM as segments. Let the {\em undiscounted contribution} $\mathcal{V}(B_i)$ of a buyer $B_i$, when offered price $x(B_i)$, be $\alpha(B_i)=x(B_i) \tp_{ix(B_i)}$, while its {\em weight} be $\tp_{ix(B_i)}$, its success probability. Undiscounted contribution $\mathcal{V}(S)$ of a segment $S$ is the sum of undiscounted contributions of buyers in the segment, and the weight of the segment is the sum of their weights. 

Given an SPM, let $dis(B)$ denote the probability that the selling process reaches buyer $B$. The {\em real contribution} of a buyer to the expected revenue is $\alpha(B) dis(B)$, and the expected revenue of the SPM is the sum of the real contributions of all the buyers. More generally, let $\gamma_\ell(B)$ denote the probability that $B_i$ is reached with at least $\ell$ items remaining. Then $dis(B) = \gamma_1(B)$. The discount factor $dis(S)$ of a segment $S$, whose first buyer is $B$, is defined to be $dis(B)$. Similarly, we define $\gamma_\ell(S)=\gamma_\ell(B)$.

We present our algorithm through a series of structural lemmas, each of which follows quite easily from the preceding lemmas.
%, so their proofs have either been moved to Appendix \ref{missing-proofs} or omitted completely. 
The first step towards our algorithm is that we can restrict our attention to truncated SPMs.

\vspace*{-0.0cm}\begin{lemma}\label{truncated}
There exists an SPM of total weight at most $K \log \frac{K}{\eps}$, where each buyer has discount factor at least $\eps$, that gives an expected revenue of at least $(1-\eps)\opt$.
%Moreover, all
%but possibly the last buyer in the sequence have success probability
%at most $(1-\eps)$.
We shall refer to SPMs that satisfy this condition as {\em truncated}.
\end{lemma}
\begin{proof}
Consider the smallest prefix of the optimal SPM (with expected revenue $\opt$) such that the discount factor of the corresponding suffix, obtained by removing the prefix, is at most $\eps$. Moreover, if we were to simply omit this prefix, then the expected revenue of the remaining segment can at most be $\opt$. So the contribution of the remaining segment to the optimal SPM is at most $\eps \opt$, and the prefix alone has expected revenue expected revenue at least $(1-\eps)\opt$. By Fact \ref{lb}, the probability that no copy gets sold in a segment of weight $\log \frac{K}{\eps}$ is at most $\eps/K$. Thus the weight of the prefix is at most $K\log \frac{K}{\eps}$. 
\end{proof}

We can now restrict ourself to approximating an optimal truncated SPM.
The following definition of a permutable segment will be crucial to the description of our algorithm.

\vspace*{-0.0cm}\begin{definition}\label{permsegment}
We shall call an SPM segment {\em permutable} if either:
\begin{enumerate}
\item its weight is at most $\delta = \frac{\eps^3}{20K^3}$. We shall refer to such a permutation segment as a {\em small buyers segment}.
\item it has a single buyer, possibly of weight more than $\delta$. In this case, we shall refer to this buyer as a {\em big buyer}.
\end{enumerate}
\end{definition}

Any SPM can clearly be decomposed into a sequence of permutable segments and big buyers. Moreover, any truncated SPM can be 
decomposed into a sequence of at most $C=O(\frac{K\log \frac{K}{\eps}}{\delta})$ permutable segments. This is because if the permutable segments are maximally chosen, then two consecutive permutable segments in the decomposition either have at least one big buyer between them, or their weights must add up to more than $\delta$ (otherwise, the two segments can be joined to create one permutable segment).

\begin{fact}\label{lb}
Let $1 > y_1,y_2\ldots y_\ell > 0$. Let $\sum_{j=1}^\ell y_j = s$. Then $1-s+s^2 > e^{-s}>\prod_{j=1}^{\ell} (1-y_j) >
1 - s$.
\end{fact}

%Proof of the following lemma is in the Appendix.
\vspace*{-0.0cm}\begin{lemma}\label{dis-eff}
The probability of selling at least one copy of the item in a small buyers permutable segment that has weight $s$ is at least $s-s^2$. The probability of selling at least $t\geq 1$ copies (assuming that at least $t$ copies are left as inventory) in such a segment is at most $s^t$. So the probability of selling exactly one copy is at least $s-2s^2$.
\end{lemma}
\begin{proof}
Fact \ref{lb} implies that the probability of selling at least one item is at least $s-s^2$ and at most $s$.

For the second statement, consider $t=2$. Conditioning on a particular buyer $B$ in the segment buying a copy, the probability that the remaining buyers in the segment buy at least one copy is at most $s$. The two events are independent, so the probability of their simultaneous occurrence is the product of their probabilities. Summing over all buyers in the segment, we get that the probability that at least two items are bought is at most $s^2$. The argument scales in a similar fashion for higher values of $t$: probability that $t$ items are bought is at most $s^t$.
\end{proof}

\vspace*{-0.0cm}\begin{lemma}\label{cont}
Consider a permutable segment of weight $s$ appearing in an SPM, and let its discount factor be $\gamma$. Then the discount factor of the last buyer in the segment is at least $\gamma(1-s)$. If the undiscounted contribution of the segment is $\alpha$, then the real contribution of buyers in this segment to the expected revenue is at least $\alpha \gamma (1-\delta)$ and at most $\alpha \gamma$.
\end{lemma}
\begin{proof}
The probability of the process not stopping inside the segment, conditioned upon reaching it, is at least the probability of not selling any copy in the segment, which is at least $1-s$ (it is exactly $1-s$ for a big buyer segment).
\end{proof}

The above lemma shows that the real contribution of a segment can be approximated by the product of its discount factor and its undiscounted contribution, which does not depend on the exact buyers, their relative ordering or prices in that segment. We next show that the discount factor of a segment, given a decomposition of an SPM into permutable segments, can also be approximated as a function of the approximate sizes of preceding segments.

\vspace*{-0.0cm}\begin{lemma}\label{approx-dis}
Given an SPM, that can be decomposed into an ordering of permutable segments $S_1,S_2\ldots$. Let $S_i$ be a small buyers segment. Let $s$ be the weight of $S_i$.

Then $ \gamma_\ell(S_i) (1-s) + \gamma_{\ell+1}(S_i) s + 4s^2 \geq \gamma_\ell(S_{i+1}) \geq \gamma_\ell(S_i) (1-s) + \gamma_{\ell+1}(S_i) s -2s^2$.
\end{lemma}
\begin{proof}
Directly using the bounds in Lemma \ref{dis-eff} to the formula:\\ \hspace*{1cm} $\gamma_\ell(S_{i+1}) = \sum_{j=\ell}^K \gamma_j(S_i) \Pr{\text{Exactly } (j-\ell) \text{copies of the item are bought by buyers in } S_i}$. 
\end{proof}

The lemma below follows easily from Lemma \ref{approx-dis}.
\vspace*{-0.0cm}\begin{lemma}\label{discount}
Given any SPM decomposed into $Q\leq C$ permutable segments $S_1,S_2\ldots$, such that the weight of $S_i$ is between $s_i + \tau$ and $s_i - \tau$ for all $1\leq i\leq n'$, where $\tau = \delta/20C$. Consider an alternate SPM (with possibly different buyers), that has $n'$ buyers, and the $i^{th}$ buyer in the segment has weight $s_i$. Let $\rho(\ell,i)$ be the probability that the $i^{th}$ buyer is reached in the alternate SPM with at least $\ell$ items remaining. Then 
$$\rho(\ell,i) - 12(\delta^2 + \tau) i) \leq \gamma_\ell(S_i) \leq \rho(\ell,i) + 12(\delta^2 + \tau) i) \enspace .$$

If the SPM is truncated, then $dis(S_i)=\gamma_1(S_i)\geq \eps$, and since $i\leq Q\leq C$, $\delta=\frac{\eps^3}{20K^3}$ and $\tau \leq \delta/20C$, so we can get a multiplicative guarantee 
%$$\rho(1,i) (1 - 12(\delta^2 + \tau)C\eps^{-1})\leq dis(S_i)\leq \rho(1,i) (1 + 12(\delta^2 + \tau)C\eps^{-1}) \enspace .$$
%
%Finally, using the fact that , we get 
$\rho(1,i) (1 - \eps)\leq dis(S_i)\leq \rho(1,i) (1 + \eps) \enspace .$
\end{lemma}

We shall refer to the following as a {\em configuration}: An ordering  of up to $C$ permutable segments, where each permutable segment is specified only by the weight of the segment and big buyer respectively, each weight being a multiple of $\tau = \frac{\delta}{20C}$. Note that the configuration does NOT specify which buyer belongs to which segment, or the individual weights of the buyers. This is because a configuration is specified by at most $C$ positive integers (weight of each segment is specified by a positive integer $z<\frac{1}{\tau}$, which indicates that the weight is $z\tau$). We shall represent a configuration $z$ as an ordered tuple of integers $(z_1,z_2,z_3\ldots)$. {\em Note that there are at most $(\frac{1}{\tau})^{O(C)} = (\frac{K}{\eps})^{O(K)}$ distinct configurations.} We say that an SPM has configuration $z$ if it can be decomposed into an ordering of permutable segments $S_1,S_2\ldots$ such that $S_i$ has weight at least $(z_i-1)\tau$ and at most $z_i\tau$.

For any given configuration $z$, the expected revenue of an SPM with configuration $z$ can be approximated, up to a factor of $(1-\delta)(1-2\eps)$ by a linear combination of the undiscounted contribution of the permutable segments, {\em where the coefficients of the linear combination depend only on $z$}. The coefficients are the discount factors, which can be computed by looking at an alternate SPM with a buyer for each segment, such that the $i^{th}$ buyer has weight $z_i\tau$. This is a direct conclusion of Lemma \ref{discount} and Lemma \ref{cont}. The discount factors of each buyer in the alternate SPM can be easily computed in $O(CK)$ time using dynamic programming. Let $A_z (i)$ denote the discount factor of the $i^{th}$ buyer in the alternate SPM corresponding to $z$. 

For any configuration $z$, we compute prices for the buyers, and a division of buyers into permutable segments $S_1,S_2\ldots$ such that $S_i$ has weight {\em at most} $z_i\tau$, and $\sum_i A_z(i) \mathcal{V}(S_i)$ is maximized (it is not necessary to include all buyers). This is precisely an instance of $\extgap$, where each buyer is an object, the different possible prices and the corresponding success probabilities create the different versions, and the sizes of the bins are given by $z$, and the feasible subsets for an object simply being that each object can get into at most one bin. This can be solved as per Lemma \ref{lem-extgap}. The solution may not saturate every bin, and hence may not actually belong to configuration $z$. However, for any two configurations  $z = (z_1,z_2,z_t)$ and $z' = (z'_1,z'_2\ldots z'_t)$, such that $z_i\leq z'_i\ \forall 1\leq i\leq t$, we have $A_z(i) > A_{z'}(i)$. So the SPM formed by concatenating $S_1,S_2\ldots$ in that order generates revenue at least $(1-3\eps)$ times the revenue of the optimal sequence that has configuration $z$. 

Thus our algorithm is to find an SPM for each configuration, using the algorithm for $\extgap$, and output the best SPM among them as the solution.

\subsection{PTAS for Computing ASPM}\label{sec-ptas-aspm}
We now design an algorithm to compute a near-optimal SPM for constant $K$. 

\vspace*{-0.0cm}\begin{theorem}\label{ptas-aspm}
There exists a PTAS for computing an optimal SPM, for any constant $K$. The running time of the algorithm is $\left(\frac{nk}{\eps}\right)^{(k\eps^{-1})^{O(k)}}$, and gives $(1-\eps)$-approximation.
\end{theorem}

As mention in Section \ref{prelim}, an ASPM is specified by a decision tree, with each node containing a buyer and an offer price. We extend some definitions used for SPMs to ASPMs. The {\em weight of a node} is the success probability at this node conditioned on being reached. A {\em segment}  in an ASPM is a contiguous part of a path (that the selling process might take) in the decision tree. A segment is called {\em non-branching} if all but possibly the last node are non-branching.  Other definitions such as weight and contribution of a segment are identical. A {\em permutation segment} is a non-branching segment satisfying properties as defined earlier (Definition \ref{permsegment}). The {\em discount factor} of a node (or a segment starting at this node, or a subtree rooted at this node) is the probability that the node is reached in the selling process.

%The algorithm is a generalization of our algorithm for SPM for constant $K$. In essence, SPMs form a special class of ASPMs where the decision tree is a single path with no branching. In principle, an ASPM can also be decomposed into permutable segments, forming a tree structure among themselves.

Consider any ASPM whose tree is decomposable into $D$ non-branching segments, each of weight at most $H$. (Note that $D=1$ for an SPM.) Then the entire tree of a truncated ASPM decomposes into $C=O(DH/\delta)$ permutable segments. We shall refer to such ASPMs as {\em $C$-truncated ASPMs}. A configuration for a $C$-truncated ASPM shall now list the weights of at most $C$ permutable segments and also specify a {\em tree structure} among them, {\em i.e.} the parent segment of each segment in the decision tree. Moreover, since each path can have no more than $C$ segments, it is sufficient to specify the weights to the nearest multiple of $\tau = \delta/20C$, to get the discount factor of each segment with sufficient accuracy. So there are $(C/\tau)^{O(C)} = C^{O(C)}$ configurations for $C$-truncated ASPMs. 

For each configuration, we can use $\extgap$ to compute an ASPM that is at least $(1-\eps)$ times the revenue of an optimal ASPM with that configuration, as before. Each $\extgap$ instance has $C$ bins in this case. The discount factor of each permutable segment in the configuration can be computed with sufficient accuracy, similar to Lemma \ref{discount}. Iterating over all possible configurations, we can find a near-optimal $C$-truncated ASPM. Solving $\extgap$ requires time exponential in the number of bins (see Lemma \ref{lem-extgap}), so the entire running time of the above algorithm is $\left(\frac{nkC}{\eps}\right)^{O(C)})$.

The problem is that for the above algorithm to be a PTAS, $C$ must be a function of $K$ and $\eps^{-1}$ only.  
Lemma \ref{truncated-aspm} achieves this goal through a non-trivial structural characterization, and immediately implies Theorem \ref{ptas-aspm}. %Its proof is provided in  Appendix \ref{missing-proofs} due to space limitations.

\vspace*{-0.0cm}\begin{lemma}\label{truncated-aspm}
There exists an ASPM with the following properties:
\begin{enumerate}
\item Its expected revenue is at least $(1-\eps)$ times the expected revenue of the optimal ASPM.
\item The decision tree is decomposable into $D=(K/\eps)^{O(K)}$ non-branching segments.
\item Each non-branching segment in the tree has weight at most $H=(K/\eps)^{O(1)}$.
\item Each path in the tree consists of at most $(K/\eps)^{O(1)}$ permutable segments.
\end{enumerate}
\end{lemma}
\begin{proof}
Let us view an optimal ASPM decision tree, with expected revenue $\opt$ as consisting of a {\em spine}, which is the path followed if no buyer buys a copy, along with decision subtrees hanging from many, possibly all, nodes of the spine. Note that all nodes may not be branching nodes, so a spine need not be left by the process at the very moment that a sale is recorded, but may {\em branch out} at a later point. Each such subtree, hanging from a node $w$ (say) on the spine, are optimal ASPMs for selling some $\ell < K$ copies to only buyers that are do not appear in any ancestor node of $w$. We shall only focus on how to modify the ASPM to have 
\begin{itemize}
\item there are at most $(K/\eps)^{O(1)}$ branching nodes on the spine, and 
\item the weight of the spine shall be at most $(K/\eps)^{O(1)}$,
\end{itemize}
while only losing a factor of $(1-c\eps)$ in expected revenue for some constant $c$.

The subtrees, since they are selling less than $K$ copies, can be transformed inductively (when a single copy is left, the subtree is just a path and trivially satisfies the required properties). Such a tree will satisfy the properties listed in Lemma \ref{truncated-aspm} (for the last property, note that any path can be decomposed into at most $K$ contiguous parts, each of which is a spines of some subtree, since leaving a spine implies a sale). 
Overall, the entire transformation shall cause a loss factor of $(1-cK\eps)$. This achieves our goal, since we could have instead started by scaling down $\eps$ to $\eps/cK$.

As a first step, we truncate the spine. For any node $w$, let $R(w)$ be the expected revenue obtained from the rest of the selling process (excluding the contribution of the buyer at $w$ itself), conditioned upon the selling process reaching node $w$. We find the earliest ({\em i.e.} closest to the root) node $w$ on the spine such that $R(w)\leq\eps\opt$, and delete all children of $w$ and the subtrees under them. This only causes a loss of $\eps \opt$ -- moreover, the probability of reaching $w$ could have been at most $\eps$, so the weight of the truncated spine is at most $K\log \frac{K}{\eps}$ (similar argument as Lemma \ref{truncated}). This immediately achieves the second property listed above, and it remains to limit the number of branching nodes. We can now assume that $R(w)> \eps \opt$ for all nodes $w$ on the spine.

For a node $w$, let $R'(w)$ denote, conditioned upon the selling process reaching $w$ and then have less than $K$ items to sell after $w$, the expected revenue from the rest of the selling process. Clearly, $R'(w)<R(w)$, since only higher revenue can be gained from the same set of buyers if there are more copies of the item to sell. A somewhat less obvious fact is that $R'(w)>R(w)/4K$. This is because $R'(w)$ is the result of selling at least one copy of the item to the same set of buyers as $R(w)$, except that $R(w)$ may have as many as $K$ copies of the item. Looking back at Section \ref{sec-LPalgo}, if the number of items is decreased from $K$ to $1$ (keeping set of buyers unchanged), then the optimum of the linear program $\lpkspp$ decreases by a factor of at most $K$ (scaling down the variables by a factor of $K$ gives a feasible solution), and the optimal revenue is always within factor $1/2$ of the LP optimum (since $1-\frac{1}{\sqrt{2\pi K}} \geq 1/2$ for all $K$). This shows that $R'(w) > \eps\opt/4K$ for all nodes on the spine.

Divide the spine into segments that either consist of a single buyer, or multiple buyers whose weights add up to no more than $\delta$. These segments may have branching nodes in them, and hence may not be permutable. Clearly there are at most $(K/\eps)^{O(1)}$ such segments, and now we shall focus on modifying each segment separately. We shall modify subtrees hanging from nodes in the segment, so that the segment can be subdivided into $poly(K/\eps)$ non-branching segments, thus completing the proof. Clearly we need to only consider those segments that comprise multiple small buyers. Let us consider one such segment, and describe the necessary modification to the tree. 

Define a minimal set of {\em pivotal nodes} in the segment, that satisfies the following condition: For any node $w$ in the segment, there is a pivotal node $v$ that is a descendant of $w$, such that $R'(v)\geq (1-\eps)R'(w)$. Since $\frac{\eps}{4K}\opt \leq R'(w)\leq \opt$ for all nodes $w$, we have at most $O(\eps^{-1} \log(K/\eps))$ pivotal nodes. We shall make modifications to the decision tree so that the {\em pivotal nodes are the only branching nodes in the segment}.

Let $v$ be the pivotal node satisfying this condition for $w$, that is nearest to $w$ in the segment. Suppose that $w$ is a branching node. We delete all children of $w$ that are not part of the spine, and simply make it a non-branching node. We do this for all non-pivotal, branching nodes in the segment. Recall that at every branching node, the choice of which children the process follows is based only upon the number of copies of the item left. Now, the segment has few enough branching nodes -- branching nodes are a subset of pivotal nodes. To argue a limited loss in revenue, we need to analyze the values $R'(w)$ in the modified trees, let us denote them by $R'_{mod}(w)$. It suffices to show that $R'_{mod}(w)\geq (1-2\eps) R'(w)$. Since $R'(w_1)$ and $R'_{w_2}$, where $w_1$ and $w_2$ are distinct nodes on the spine, are expectations conditioned upon disjoint events, this implies that the expected revenue of the entire tree falls by a factor of at most $(1-2\eps)$ due to this modification.

To show that $R'_{mod}(w)\geq (1-2\eps) R'(w)$, we can almost say that $R'_{mod}(w)$ is at least to $R'(v)$, since the branching has been deferred until node $v$. The only difference is that some small buyers get executed between $w$ and $v$. So if there are $\ell$ items left after $w$, there may be less than $\ell$ items when $v$ is reached in the modified tree -- however, the probability of this event is less than $\delta$, and is independent of the history of events up to $w$. So, neglecting the contribution of nodes between $w$ and $v$ (but taking into account their discounting effect on descendant nodes), $R'_{mod}(w) \geq (1-\delta) R'(v)$. Since $R'(v)\geq (1-\eps)R'(w)$, we have our result.

Thus each segment has at most $O(\eps^{-1} \log(K/\eps))$ branching nodes now, which implies that the entire spine has $(K/\eps)^{O(1)}$ branching nodes. This completes the proof.
\end{proof}

\bibliographystyle{abbrv}
\bibliography{pricingbiblio} %pricingbiblio.bib is the name of the Bibliography in this case

\appendix

\section{Discretization}\label{discrete}

We explain why we can assume the following for the value distribution of each buyer: it is discrete, and the probability mass at all points, if non-zero, is an integer multiple of $\frac{1}{n^2}$. The assumption can only cause a loss of $(1-\frac{1}{n})$ in the expected revenue: given an instance, we can create a discrete distribution with the above properties, corresponding to each value distribution, and an algorithm for computing an $\alpha$-approximate SPM or ASPM in the modified instance gives an $\alpha(1-\frac{1}{n})$-approximation for the original instance.

Let $c_{iv}=v\tp_{iv}$ be the expected revenue from buyer $B_i$ if price $v$ is posted to it. First, we can simply keep only those $v$ that are powers of $(1-\frac{1}{n^2})$, and assume that there is probability mass on only these points (leave $\tp_{iv}$ unchanged). Next, for each such $v$, alter $v$ and $\tp_{iv}$ so that their product $c_{iv}$ remains unchanged, but $\tp_{iv}$ changes to the closest integral multiple of $\frac{1}{n^2}$ that is greater than $\tp_{iv}$. This does not change the possible choices of expected revenue that can be obtained from a buyer upon reaching it, and their effect on future buyers, {\em i.e.} success probability, changes by $1/n^2$. The changes in the effect on the future can add up over $n$ buyers to change the probability of reaching a particular buyer by at most $1/n$, so we can neglect this change.

\section{A Property of Poisson Distribution}\label{poisson}

The proof of Lemma \ref{approxfactor} uses the following property of Poisson variables.
\vspace*{-0.0cm}\begin{lemma}
Let $P$ be a Poisson variable with mean $K$, {\em i.e.} for all integers $m\geq 0$, $\Pr{P=m} = \frac{K^m}{m! e^K}$. Then $\E{\max\{0,P-K\}} = \frac{K^{K+1}}{K! e^K}$, and so $\E{\min\{P,K\}} = \E{P - \max\{0,P-K\}} = K - \frac{K^{K+1}}{K! e^K}$.
\end{lemma}
\vspace*{-0.0cm}\begin{proof}
All we need to show is that $\sum_{m=K+1}^{\infty} \frac{K^m (m-K)}{m!} = \frac{K^{K+1}}{K!}$.

It is easy to show by induction that for any $j\geq 1$, 
$$\frac{K^{K+1}}{K!} - \sum_{m=K+1}^{K+j} \frac{K^m (m-K)}{m!} = \frac{K^{K+j+1}}{(K+j)!}\enspace .$$

Since $\lim_{x\ra \infty} \frac{K^{x+1}}{x!} =0$, the proof is complete.
\end{proof}

%The following facts will be useful in providing the missing proofs.
%
%\begin{fact}\label{lb}
%Let $1 > y_1,y_2\ldots y_\ell > 0$. Then $e^{-\sum_{j=1}^\ell y_j}>\prod_{j=1}^{\ell} (1-y_j) >
%1 - \sum_{j=1}^\ell y_j$.
%\end{fact}
%
%\begin{fact}\label{small}
%Let $\delta > y_1,y_2\ldots y_\ell > 0$. Then $e^{-\sum_{j=1}^\ell
%  y_j} e^{- \delta \sum_{j=1}^\ell y_j} < \prod_{j=1}^{\ell} (1-y_j) <
%e^{-\sum_{j=1}^\ell y_j}$ (since $(1-y_j) > e^{-y_j - y_j^2}$).
%\end{fact}
%
%\begin{fact}\label{ineq1}
%Let $y<1-\eps$ for some $\eps>0$. Then, for any $\eps> \tau>0$, $1 - y
%+ \tau < (1+\frac{\tau}{\eps})(1-y)$.
%\end{fact}
%\vspace*{-0.0cm}\begin{proof}
%The right hand side is $(1+\frac{\tau}{\eps})(1-y) = 1-y + \frac{\tau
%  (1-y)}{\eps}$. Since $1-y > \eps$, the right hand side is more than
%$1-y+\tau$.
%\end{proof}

%\begin{fact}\label{small}
%Let $\delta > y_1,y_2\ldots y_\ell > 0$. Then $e^{-\sum_{j=1}^\ell
%  y_j} e^{- \delta \sum_{j=1}^\ell y_j} < \prod_{j=1}^{\ell} (1-y_j) <
%e^{-\sum_{j=1}^\ell y_j}$ (since $(1-y_j) > e^{-y_j - y_j^2}$).
%\end{fact}

\end{document}